\documentclass[a4paper,10pt]{article}
\usepackage[utf8]{inputenc}
\usepackage{todonotes}

\usepackage{amsthm,amssymb,amsmath}

\newtheorem{lemma}{Lemma}
\newtheorem{corollary}{Corollary}
\newtheorem{theorem}{Theorem}

\title{A short note on the counting complexity of conjunctive queries}
\author{Stefan Mengel\\ Univ. Artois, CNRS, \\Centre de Recherche en Informatique de Lens (CRIL),\\ F-62300 Lens, France}

\begin{document}

\maketitle

\begin{abstract}
 This note closes a minor gap in the literature on the counting complexity of conjunctive queries by showing that queries that are not free-connex do not have a linear time counting algorithm under standard complexity assumptions. More generally, it is shown that the so-called quantified star size is a lower bound for the exponent in the runtime of any counting algorithm for conjunctive queries.
\end{abstract}

\section{Introduction}

Conjunctive queries are the most basic and arguably most important class of database queries~\cite{ChandraM77}. As such, they have been one of the major objects of study of database theory and are by now well understood in many respects. 

In this short note, we close a minor gap in the literature of the counting complexity of conjunctive queries. This area has seen intensive study in the last few years, see e.g.~\cite{PichlerS13,BraultBaron13,DurandM14,DurandM15,GrecoS14,ChenM15,DellRW19} for a sample. The main difference between the counting complexity of CQs and the decision problem for Boolean CQs is that projection, also called existential quantification, matters for the complexity of queries: it was already observed by Pichler and Skritek in~\cite{PichlerS13} that there are classes of acyclic CQs for which counting answers is $\mathsf{\#P}$-hard whereas Yannakakis' algorithm shows that Boolean CQs can be decided efficiently. The understanding of counting was refined in~\cite{DurandM14} by introducing a notion called \emph{quantified star size} that is shown for acyclic queries to determine the (parameterized) complexity of answer counting. Note that quantified star size can be seen as a quantitative generalization of free-connex acyclicity, previously introduced in the setting of enumeration algorithms~\cite{BaganDG07}, in the sense that the free-connex acyclic queries are exactly the acyclic queries of quantified star size $1$.

The work of Dell, Roth and Wellnitz~\cite{DellRW19} gave a sweeping generalization of~\cite{DurandM14}, in particular introducing the perspective of fine-grained complexity in the setting which allows to give tight results on the exponent of the polynomial runtimes. This in particular allows considering the complexity of individual queries instead of whole classes of them as they have to be considered in the parameterized complexity approach. Concretely, \cite{DellRW19} shows that a parameter called \emph{dominating star size} lower bounds the exponent of counting algorithms for CQs. 

The very interesting results of~\cite{DellRW19} have two slight shortcomings: first, the complexity is measured in the size of the domain of the input database. This is very common in many areas of theoretical computer science, in particular most work on graph algorithms does this. However, in database theory it is common to measure complexity with respect to the size of the database instead. It is not hard to see that the dominating star size of~\cite{DellRW19} does not determine the counting complexity of CQs when measured in the size of the input database which is a slight mismatch between~\cite{DellRW19} and most of the literature in database theory. Second, the results of~\cite{DellRW19} only consider cases in which the dominating star size of the query is at least $3$. This in particular makes it impossible to determine the queries in which linear time counting is possible. This is unfortunate since, due to very big input databases, differentiating between problems that have linear time algorithms and such that have not is often particularly interesting in database theory, see e.g.~the great amount of work in this direction on enumeration algorithms~\cite{BerkholzGS20}.

In this short note, we slightly improve some of the results of~\cite{DellRW19} in the directions pointed out above: we show that for acyclic CQs the original star size as introduced in~\cite{DurandM14} is a lower bound for the exponent of the runtime for counting when measuring in the size of the database. In particular, this is also true for quantified star size $1$ and $2$ which allows us to show that the free-connex CQs are the only self-join free acyclic CQs which allow linear time counting. Let us stress that the the techniques used to show these results are neither new nor surprising; essentially, we tie together several techniques from the literature here.

\section{Preliminaries}

We assume that the reader is familiar with the basics of database theory~\cite{AbiteboulHV95}. We only consider conjunctive queries which we generally assume to be self-join free. The hypergraph of a conjunctive query $q$ has as vertices the variables of $q$ and an edge $e$ for every atom which contains exactly the variables appearing in the atom.

A \emph{join tree} of a hypergraph $H$ is a tree $T$ whose vertices are the edges of $H$ and which has the following property: for every $v\in V(H)$ the set $\{e\in E(H) \mid v\in e\}$ is connected in $T$.
A hypergraph $H$ is called \emph{acyclic} if it has a join tree. A CQ is called acyclic if its hypergraph is acyclic.

All graphs in this note are simple and undirected.
A dominating set $S$ of a graph $G=(V,E)$ is a set $S\subseteq V$ such that every vertex not in $S$ has a neighbor in $S$. The problem $k$-Dominating Set (short $k$-DS) is to decide, given a graph $G$, if $G$ has a dominating set of size at most $k$. We will use the following result from~\cite{PatrascuW10}.

\begin{theorem}\label{thm:PatrascuW10}
 If SAT has no algorithm with runtime $O(2^{n(1-\epsilon)})$ for any $\epsilon>0$, then there is no constant $\epsilon'$ such that there is a constant $k$ and an algorithm for $k$-DS with runtime $O(n^{k-\epsilon'}) $ on graphs with $n$ vertices.
\end{theorem}

We remark that the assumption on SAT in Theorem~\ref{thm:PatrascuW10} is in particular implied by the strong exponential time hypothesis (SETH)~\cite{ImpagliazzoP01}. As a consequence, our results below could also be with SETH. However, we keep the formulation as in~\cite{PatrascuW10}.

\section{Star Queries}

In this section, we will refine the lower bounds for the star queries that have already been a crucial building block of the lower bounds in~\cite{DurandM14,DurandM15,DellRW19}.

\begin{lemma}\label{lem:stars}
 Let $k\in \mathbb{N}$ with $k\ge 2$. If there is an algorithm that counts answers to \[q^\star_k(x_1, \ldots, x_k) := \exists z \bigwedge_{i\in [k]} R(x_i,z)\] in time $O(m^{k-\epsilon})$ on databases with $m$ tuples for some $\epsilon>0$, then there is a $k'\in \mathbb{N}$ such that $k'$-DS can be decided in time $O(n^{k'-\epsilon})$ on graphs with $n$ vertices.
\end{lemma}

We remark that a very similar lower bound has already been shown in~\cite{DellRW19} but there the lower bound is for $n^{k-\epsilon}$ where $n$ is the size of the domain and not $m^{k-\epsilon}$. Moreover, there is was assumed that $k\ge 3$. We slightly improve the bound here by encoding several vertices of the graph in each variable of the query instead of just encoding a single vertex per variable.

\begin{proof}[Proof of Lemma~\ref{lem:stars}]
 Choose $k'$ as a fixed integer such that $k'>k^2/\epsilon$ and $k'$ is divisible by $k$. We will encode $k'$-DS into the query $q^\star_k$. To this end, let $G=(V, E)$ be an input for $k'$-DS. Set
 \[R:=\{(\vec{u}, v)\mid v\in V, \vec{u} = (u_1,\ldots, u_{k'/k}), \forall i \in [k'/k]: u_iv\notin E, u_i \ne v\}.\]
 Then any assignment $\vec{u}^1, \ldots, \vec{u}^k$ to $x_1, \ldots, x_k$ in $q(x,y)$ corresponds to a choice $S$ of at most $k'$ vertices in $G$. The set $S$ is a dominating set if and only if there is no vertex $v$ in $V$ that is not in $S$ and has no neighbor in $S$. This is the case if and only if there is no $v\in V$ that assigned to $z$ makes $q^\star_k$ true. Thus the answers to $q^\star_k(x_1, \ldots, x_k)$ are exactly the assignments $\vec{u}^1, \ldots, \vec{u}^k$ that do \emph{not} correspond to dominating sets in $G$. Hence, any algorithm counting the answers to $q^\star_k(x_1, \ldots, x_k)$ directly yields an algorithm for $k'$-DS.
 
 We now analyze the runtime of the above algorithm. Note that the time for the construction of $R$ is negligible, so the runtime is essentially that of the counting algorithm for $q^\star_k(x_1, \ldots, x_k)$. First observe that the relation $R$ has at most $n^{\frac{k'}{k}+1}$ tuples. The exponent of the runtime of the counting algorithm is thus
 \begin{align*}
  \left(\frac{k'}{k}+1\right)(k-\epsilon) &= k'+k-\frac{k' \epsilon}{k} - \epsilon\\
  & < k'+k - \frac{k^2 \epsilon}{\epsilon k} -\epsilon\\
  & = k'-\epsilon
 \end{align*}
where the inequality comes from the choice of $k'$ satisfying $k'> k^2/\epsilon$.
\end{proof}

\begin{corollary}
  If SAT has no algorithm with runtime $O(2^{n(1-\epsilon)})$ for any $\epsilon>0$, then there is no constant $k\in \mathbb{N}$, $k\ge 2$ and no $\epsilon' > 0$ such that there is an algorithm that counts answers to \[q^\star_k(x_1,\ldots, x_k) := \exists z \bigwedge_{i\in [k]}R(x_i,z)\] in time $O(m^{k-\epsilon'})$ on databases with $m$ tuples. 
\end{corollary}

\section{Tight Bounds for Acyclic Queries}

In this section, we will lift Lemma~\ref{lem:stars} to self-join free acyclic queries. The approach is rather standard, see e.g.~\cite{BraultBaron13,DurandM14,DellRW19}, but we here consider somewhat tighter runtime bounds so we give it here for completeness. We closely follow the theory as developed in~\cite{BraultBaron13}. The argument is a minimal adaption of one found in the (unfortunately unpublished) full version of~\cite{BaganDG07} for enumeration. An acyclic conjunctive query with hypergraph $\mathcal{H}$ and free variables $S$ is called \emph{free-connex} if the hypergraph $\mathcal{H}\cup \{S\}$ that we get from $\mathcal{H}$ by adding $S$ as an edge is acyclic as well.

\begin{theorem}\label{thm:embedding}
 Let $q$ be a self-join free conjunctive query that is acyclic but not free-connex. Then, assuming that SAT has no algorithm with runtime $O(2^{n(1-\epsilon)})$ for any $\epsilon>0$, there is no algorithm that counts the solutions of $q$ on a database with $m$ tuples in time $O(m^{2-\epsilon'})$ for any $\epsilon'$. 
\end{theorem}
\begin{proof}
 The idea is to embed $q^\star_2$ into $q$. Let $S$ be the set of free variables of $q$ and let $\mathcal{H}$ be the hypergraph of $q$. Since $q$ is acyclic but not free-connex, we get by \cite[Theorem 13]{BraultBaron13} that there is a set $S'$ such that $(\mathcal{H}\cup \{S\})[S']$ is a cycle plus potentially some unary edges. Note that acyclic hypergraphs cannot contain such an induced cycle, so we it follows that $(\mathcal{H}\cup \{S\})[S']$ contains at least one edge $xx'$ which is not in $\mathcal{H}[S']$. Moreover, since the cycle is induced, it must contain a vertex not in $S$. It follows that we can choose a path $P$ in $\mathcal{H}[S']$ such that the endpoints of $P$ are in $S$ but none of the internal vertices are. Let $x_1, x_2$ be the endpoints and let $z_1, \ldots, z_\ell$ be the internal vertices.
 
 Let $\mathbb{D}$ be an input database for $q^\star_2$. We construct a database $\mathbb{D}'$ for $q$ as follows: all vertices not on $P$ are forced to take a fixed value $d$ in all relations. If there are any unary atoms $U$ on any variable on $P$, we make sure that those do not constrain the variables: if the atom is in the variable $x_i$, we set $U^{\mathbb{D}'}$ to the active domain of $x_i$ in $\mathbb{D}$. If the atom is in any $z_i$, we set $U^{\mathbb{D}'}$ to the active domain of $z$ in $\mathbb D$. Moreover, we construct the following:
 \begin{itemize}
 \item For every atom $R_1$ containing $x_1$ and $z_1$, we construct a relation $R_1$ such that the projection $\Pi_{x_1, y_1}(R_1)$ is $R^{\mathbb{D}}$, the relation for $R(x_1, z)$ in $q^\star_2$. 
 \item We set the relations for atoms containing $x_2$ and $z_\ell$ analogously. 
 \item Finally, for every atom of $q$ that contains a pair $z_i, z_{i+1}$, we construct a relation $R$ such that the projection $\Pi_{z_i, z_{i+1}}(R)= \{(v,v)\mid v\in D_z\}$ where $D_z$ is the active domain of $z$ in $\mathbb{D}$.
 \end{itemize}
Note that, for every atom of $q$ containing a variable on $P$, exactly one of the cases applies. It follows that $\mathbb{D}'$ has linear size in $\|\mathbb D\|$ and that it can be constructed in quasilinear time. Also, observe that $\Pi_{x_1, x_2}(q(\mathbb{D}')) = q^\star_2(\mathbb{D})$ and that this projection is a bijection. It follows that $|q(\mathbb{D}')| = |q^\star_2(\mathbb{D})|$, so counting for $q$ directly allows counting for $q^\star_2$.
\end{proof}

Note that it is well known that for every free-connex acyclic CQ the answers can be counted in linear time, see e.g.~\cite{BraultBaron13,CarmeliZBKS20}, so we get that being free-connex completely characterizes the linear time case for acyclic CQs in the following sense.

\begin{corollary}
 Assume that SAT has no has no algorithm with runtime $O(2^{n(1-\epsilon)})$ for any $\epsilon>0$. Let $q$ be a self-join free acyclic CQ. Then there is an algorithm that counts the answers of $q$ for every database in linear time if and only if $q$ is free-connex.
\end{corollary}

We remark in passing that Theorem~\ref{thm:embedding} can be generalized beyond quadratic runtime lower bounds. There is a quantitative generalization of being non-free-connex called \emph{quantified star size} that essentially measures under which conditions $q^\star_k$ can be embedded into a query. Inspection of the proof of Lemma~4.3 from~\cite{DurandM15} shows that the following is true (we refer to~\cite{DurandM15} for definitions and all details).

\begin{corollary}
 Let $q$ be a self-join free conjunctive query of quantified star size $k$. Then, assuming that SAT has no algorithm with runtime $2^{n(1-\epsilon)}$ for any $\epsilon>0$, there is no algorithm that counts the solutions of $q$ on a database with $m$ tuples in time $m^{k-\epsilon'}$ for any $\epsilon'$. 
\end{corollary}

\section{Conclusion}

We have slightly improved some of the bounds from~\cite{DellRW19} here by improving the dependence of the runtime bounds from the size of the domain to the size of the database. This makes these bounds connect better to other results in database theory. Let us stress that we have only considered a very small portion of the results in the very rich paper~\cite{DellRW19}. It would be interesting to see which other results from there can be adapted in this way.

One compelling question that we leave open is fully determining the queries for which---under a standard complexity assumption---answers can be counted in linear time. Note that, combining our result with the work of Brault-Baron in~\cite{BraultBaron13}, this question reduces to showing lower bounds for the so-called Loomis-Whitney joins. Brault-Baron conjectured that those queries do not have a linear time counting algorithms and from this and our result it would follow that the queries that allow linear time counting are exactly the free-connex acyclic queries of~\cite{BaganDG07}. We note that there are reasons to consider Brault-Baron's conjecture as plausible, see the discussion in~\cite[Section~6]{BerkholzGS20}. However, it would still be interesting to find more evidence for it.
Unfortunately, doing so looks quite challenging: even for the question of counting triangles in a graph, the smallest Loomis-Whitney join, there are no known conditional lower bounds under standard assumptions, despite the fact that triangles play a major role in fine-grained complexity, see e.g.~\cite{williams2018some}.

\paragraph{Acknowledgement.} This note would not have been written without a question Nicole Schweikardt and it would not have been published without another question by Nofar Carmeli. The author thanks Nofar Carmeli and Nikolaos Tziavelis for helpful discussions and their generous advice that helped greatly improve upon an earlier version of this note.

\bibliographystyle{plain}
\bibliography{lintime}
\end{document}